\documentclass{article}

\usepackage{graphicx}
\usepackage{xcolor}

\usepackage[symbol]{footmisc}

\usepackage{amssymb}
\usepackage{eucal}
\usepackage{amsmath}
\usepackage{amsthm}
\usepackage{graphicx}
\usepackage{float}
\usepackage{framed}
\usepackage{enumerate}
\usepackage{hyperref}
\usepackage[toc,page]{appendix}

\newcommand{\wt}{\widetilde}

\newcommand{\R}{{\mathbb  R}}  \numberwithin{equation}{section} \newtheorem{thm}{\bf Theorem}[section]
\newtheorem{lem}[thm]{\bf Lemma}  
\newtheorem{cor}[thm]{\bf Corollary}

\DeclareMathOperator{\tr}{tr}

\usepackage{amssymb}

\usepackage[figuresright]{rotating}

\begin{document}

\title{Laplace-Beltrami operator on the orthogonal group in Cartesian coordinates}

\author{Petre Birtea, Ioan Ca\c su, Dan Com\u{a}nescu\\
{\small Department of Mathematics, West University of Timi\c soara} \\
{\small Bd. V. P\^ arvan, No 4, 300223 Timi\c soara, Rom\^ania}\\
{\small Email: petre.birtea@e-uvt.ro, ioan.casu@e-uvt.ro, dan.comanescu@e-uvt.ro}}

\date{}

\maketitle

\begin{abstract}
Using the embedded gradient vector field method (see P. Birtea, D. Com\u anescu, Hessian operators on constraint manifolds, J Nonlinear Sci 25, 2015), we present a general formula for the Laplace-Beltrami operator defined on a constraint manifold, written in the ambient coordinates. 
Regarding the orthogonal group as a constraint submanifold of the Euclidean space of $n\times n$ matrices, we give an explicit formula for the Laplace-Beltrami operator on the orthogonal group using the ambient Euclidean coordinates. 
We apply this new formula for some relevant functions. 
\end{abstract}

\noindent {\bf Keywords:} 
Laplace-Beltrami operator; orthogonal group; constraint manifold.\\
{\bf MSC Subject Classification:} 15B10, 53Bxx,58Cxx





\section{Introduction}

In the last decades, there was a great interest in the study of the Laplace operator on Riemannian manifolds (frequently named Laplace-Beltrami operator), see \cite{shubin}, \cite{berger}, \cite{bakry}, \cite{bakry-3}, \cite{olafsson}.

The purpose of this section is to give a formula for the Laplace-Beltrami operator on a manifold that is the preimage of a regular value for a set of constraint functions. We will apply this formula for the particular cases of the sphere and of the orthogonal group, especially by computing the Laplacian for certain functions that appear in the analysis of the heat kernel on the orthogonal group (\cite{levy}, \cite{fulman}).

Let $(M,{\bf g})$ be a Riemannian manifold of dimension $m$, $(u_1,\dots,u_m)$ a local system of coordinates and $\frac{\partial }{\partial u_1},\dots, \frac{\partial }{\partial u_m}$ the corresponding local frame. For $f:M\to \R$ a smooth function, the Laplace-Beltrami operator is defined by the formula\footnote{By convention, we take the sign $+$ in the definition of the Laplace-Beltrami operator.} (see \cite{abraham}, \cite{berger})
\begin{equation*}
\Delta_Mf({\bf u})=\tr\left([{\bf g}({\bf u})]^{-1}[\text{Hess}f]({\bf u})\right),
\end{equation*}
where $[{\bf g}({\bf u})]$ is the associated matrix of the Riemannian metric ${\bf g}$ in the local coordinates  and $[\text{Hess}f]({\bf u})$ is the Hessian matrix. 

Let $F=(F_1,\dots,F_k):M\to \R^k$ be a smooth set of constraint functions, with $k<m$, and $S_{\bf c}:=F^{-1}({\bf c})$, where ${\bf c}$ is a regular value for $F$. The induced Riemannian structure is $(S_{\bf c},{\bf g}_{\bf c})$ with ${\bf g}_{\bf c}={\bf g}\left|_{_{TS_{\bf c}\times TS_{\bf c}}}\right.$. 

We choose $\{{\bf t}_1,\dots,{\bf t}_{m-k}\}$ an adapted local frame on $S_{\bf c}$, i.e. ${\bf t}_i\in \mathcal{X}(M)$ such that ${\bf t}_i({\bf u})\in T_{\bf u}S_{\bf c}$, for all ${\bf u}$ in the domain of the system of local coordinates. Regarding ${\bf t}_i({\bf u})$ as column vectors expressed in the local frame $\frac{\partial }{\partial u_1},\dots, \frac{\partial }{\partial u_m}$, we obtain the transformation matrix 
\begin{equation}\label{transformation-matrix-11}
T:=[{\bf t}_1, \dots, {\bf t}_{m-k}]\in \mathcal{M}_{m\times (m-k)}(\R). 
\end{equation}
The matrix form for the induced Riemannian metric ${\bf g}_{\bf c}$ is 
\begin{equation}\label{g_c}
[{\bf g}_{\bf c}]=T^t[{\bf g}]T.
\end{equation}

For $\widetilde{f}:S_{\bf c}\to \R$ a smooth function, the Laplace-Beltrami operator is given by
\begin{equation}\label{Laplacian-Sc}
\Delta_{S_{\bf c}}\widetilde{f}({\bf u})=\tr\left([{\bf g}_{\bf c}({\bf u})]^{-1}[\text{Hess}_{S_{\bf c}}\widetilde{f}]({\bf u})\right),\,\,\forall{\bf u}\in S_{\bf c}.
\end{equation}
The above formula can be written in the ambient geometry of the manifold $M$. 
Let $f:M\to \R$ be a smooth prolongation of $\widetilde{f}$, i.e. $\widetilde{f}=f_{|S_{\bf c}}$, then (see \cite{Birtea-Comanescu-Hessian})
\begin{equation*}
\text{Hess}_{S_{\bf c}}\widetilde{f}=\text{Hess}f_{_{|TS_{\bf c}\times TS_{\bf c}}}-\sum_{\alpha=1}^k\sigma_{\alpha}\text{Hess}F_{{\alpha}|{_{|TS_{\bf c}\times TS_{\bf c}}}}
\end{equation*}
or equivalently, in matrix form,
\begin{equation}\label{HessT}
\left[\text{Hess}_{S_{\bf c}}\widetilde{f}\right]=T^t\left[\text{Hess}f\right]T-\sum_{\alpha=1}^k\sigma_{\alpha}T^t\left[\text{Hess}F_{\alpha}\right]T.
\end{equation}
The functions $\sigma_{\alpha}$ are the so called Lagrange multiplier functions, defined in \cite{Birtea-Comanescu-Hessian}, and are given by 
\begin{equation}\label{sigma-101}
\sigma_{\alpha}({\bf u}):=\frac{\det \left(\text{Gram}_{(F_1,\ldots ,F_{\alpha-1},f, F_{\alpha+1},\dots,F_k)}^{(F_1,\ldots , F_{\alpha-1},F_\alpha, F_{\alpha+1},...,F_k)}({\bf u})\right)}{\det\left(\text{Gram}_{(F_1,\ldots ,F_k)}^{(F_1,\ldots ,F_k)}({\bf u})\right)},
\end{equation}
where
\begin{equation*}\label{sigma}
\text{Gram}_{(g_1,...,g_s)}^{(f_1,...,f_r)}=\left[%
\begin{array}{cccc}
  {\bf g}(\nabla_{\bf g} g_1,\nabla_{\bf g}f_{1}) & ... & {\bf g}(\nabla_{\bf g} g_s,\nabla_{\bf g} f_{1}) \\
  \vdots & \ddots & \vdots \\
  {\bf g}(\nabla_{\bf g} g_1,\nabla_{\bf g}f_r) & ... & {\bf g}(\nabla_{\bf g} g_s,\nabla_{\bf g} f_r) 
\end{array}%
\right],
\end{equation*}
where $\nabla_{\bf g} h$ denotes the Riemannian gradient of the function $h$ with respect to the Riemannian metric ${\bf g}$.

As a consequence of formulas \eqref{g_c}, \eqref{Laplacian-Sc}, and \eqref{HessT}, we obtain the following result.
\begin{thm}\label{teorema-generala}
Choosing an adapted local frame  $\{{\bf t}_1,\dots,{\bf t}_{m-k}\}$, the Laplace-Beltrami operator on a constraint manifold $S_c$, written in the coordinates of the ambient space $M$, has the formula
\begin{equation*}
\Delta_{S_{\bf c}}\widetilde{f}=\tr\left((T^t[{\bf g}]T)^{-1}T^t[\emph{Hess}{f}]T\right) - \sum_{\alpha=1}^k\sigma_{\alpha}\tr\left((T^t[{\bf g}]T)^{-1}T^t\left[\emph{Hess}F_{\alpha}\right]T\right).
\end{equation*}
For the particular case when the ambient manifold is the Euclidean space, $[{\bf g}]=\mathbb{I}_m$, we have the formula
\begin{equation}\label{laplace-general-101}
\Delta_{S_{\bf c}}\widetilde{f}=\tr\left(T^+[\emph{Hess}{f}]T\right) - \sum_{\alpha=1}^k\sigma_{\alpha}\tr\left(T^+\left[\emph{Hess}F_{\alpha}\right]T\right),
\end{equation}
where the matrix $T^+:=(T^tT)^{-1}T^t$ is the (left) Moore-Penrose inverse of the matrix $T$.
\end{thm}

\section{Laplace-Beltrami operator on the sphere}

According to \eqref{laplace-general-101}, in order to write down the formula for the Laplace-Beltrami operator on the sphere $S^{n-1}_R:=\{{\bf x}\in \R^n\,|\,\sum\limits_{i=1}^nx_i^2=R^2\}$ using Euclidean coordinates, we need to introduce a local frame on the sphere. Inspired by the construction of a local frame on the Stiefel manifold, see \cite{second-order}, page 11, we consider the local frame as follows\footnote{We denote by ${\bf e}_1,\dots,{\bf e}_p$ the canonical basis in $\R^p$ and with $\left<\cdot,\cdot\right>$ the scalar product on $\R^p$.}: for ${\bf x}\in S^{n-1}_R$, with $x_j\neq 0$,  
\begin{equation}\label{local-frame-sphere}
{\bf t}_i=R^2{\bf e}_i-x_i{\bf x},\,\,\,i\in \{1,\dots,n\}\backslash\{j\}.
\end{equation}
In this case, the transformation matrix $T$ introduced in \eqref{transformation-matrix-11} is a $n\times (n-1)$ matrix. 


We denote ${\bf x}_{\widehat{\jmath}}:=(x_1,\dots,\widehat{x_{j}},\dots,x_n)^t$ the column matrix obtained from ${\bf x}$ by removing the coordinate $x_j$.

\begin{lem}\label{lema1} For the local frame chosen in \eqref{local-frame-sphere}, we have:
\begin{enumerate}[(i)]
\item 
$(T^tT)^{-1}{=\dfrac{1}{R^4}\left(\mathbb{I}_{n-1}+\dfrac{1}{x_{j}^2}{\bf x}_{\widehat{\jmath}}\,{\bf x}^t_{\widehat{\jmath}}\right)}.$
\item $TT^+=\mathbb{I}_n-\dfrac{1}{R^2}{\bf x}\,{\bf x}^t$.
\end{enumerate}
\end{lem}

\begin{proof} 
$(i)$ We denote\footnote{The matrix ${\bf x} {\bf x}^t$ is the usual matrix multiplication between a $n\times 1$ matrix and a $1\times n$ matrix.} $P:=\mathbb{I}_n-\frac{1}{R^2}{\bf x}{\bf x}^t$ and $J:=\left[{\bf e}_1,\dots,\widehat{{\bf e}_j},\dots, {\bf e}_n\right]$ ($J$ is the $n\times (n-1)$ matrix obtained from $\mathbb{I}_n$ by removing column $j$). We notice that $T=R^2PJ$ and, also, that
$$P^2=\mathbb{I}_n-\frac{2}{R^2}{\bf x}{\bf x}^t+\frac{1}{R^4}({\bf x}{\bf x}^t)({\bf x}{\bf x}^t)
=\mathbb{I}_n-\frac{2}{R^2}{\bf x}{\bf x}^t+\frac{1}{R^4}({\bf x}^t{\bf x})({\bf x}{\bf x}^t)
=\mathbb{I}_n-\frac{2}{R^2}{\bf x}{\bf x}^t+\frac{1}{R^2}({\bf x}{\bf x}^t)=P,
$$
hence $P$ is a projector matrix.\\
It is easy to see that
$$J^tJ=\mathbb{I}_{n-1}; ~~J^t{\bf x}={\bf x}_{\widehat{\jmath}},~{\bf x}^tJ={\bf x}_{\widehat{\jmath}}^t~.$$
It follows that
\begin{align*}
T^tT&=(R^2J^tP)(R^2PJ)=R^4J^tPJ=R^4J^t\left(\mathbb{I}_n-\frac{1}{R^2}{\bf x}{\bf x}^t\right)J\\
&=R^4\left(J^tJ-\frac{1}{R^2}(J^t{\bf x})({\bf x}^tJ)\right)=R^4\left(\mathbb{I}_{n-1}-\frac{1}{R^2}{\bf x}_{\widehat{\jmath}}\,{\bf x}^t_{\widehat{\jmath}}\right).
\end{align*}
Using the Sherman-Morrison formula (see, for example, \cite{gentle}, p. 221) we obtain
\begin{align*}(T^tT)^{-1}&=\frac{1}{R^4}\left(\mathbb{I}_{n-1}+\frac{1}{1-\frac{1}{R^2}{\bf x}^t_{\widehat{\jmath}}\,{\bf x}_{\widehat{\jmath}}}\cdot \frac{1}{R^2}{\bf x}_{\widehat{\jmath}}\,{\bf x}^t_{\widehat{\jmath}}\right)\\
&=\frac{1}{R^4}\left(\mathbb{I}_{n-1}+\frac{1}{R^2-(R^2-x_j^2)}{\bf x}_{\widehat{\jmath}}\,{\bf x}^t_{\widehat{\jmath}}\right)=\dfrac{1}{R^4}\left(\mathbb{I}_{n-1}+\dfrac{1}{x_{j}^2}{\bf x}_{\widehat{\jmath}}\,{\bf x}^t_{\widehat{\jmath}}\right).\end{align*}
\noindent $(ii)$ If we denote $Q:=\mathbb{I}_{n-1}+\dfrac{1}{x_{j}^2}{\bf x}_{\widehat{\jmath}}\,{\bf x}^t_{\widehat{\jmath}}~$, we have obtained in $(i)$ that $$\mathbb{I}_{n-1}= (T^tT)\left(\dfrac{1}{R^4}Q\right) =(R^4J^tPJ)\left(\dfrac{1}{R^4}Q\right)=J^tPJQ.$$ If we consider the matrix
$$V:=PJQJ^t,$$
we easily get 
$$V^2=(PJQJ^t)(PJQJ^t)=PJQ(J^tPJQ)J^t=PJQ\mathbb{I}_{n-1}J^t=V,$$
therefore $V$ is also a projector matrix. Additionally, from $J^tV=J^t$, we have that the range of $V$ is the same as the range of $P$ (since $\hbox{range}(V)\subseteq \hbox{range}(P)$ and $\hbox{rank}(V)=\hbox{rank}(P)=n-1$), leading to $VP=P$ (see, for example, \cite{yanai}, p. 35). Then,
$$TT^+=T(T^tT)^{-1}T^t=(R^2PJ)\left(\dfrac{1}{R^4}Q\right)(R^2J^tP)=(PJQJ^t)P=VP=P=\mathbb{I}_n-\frac{1}{R^2}{\bf x}{\bf x}^t,$$
which completely proves the Lemma.
\end{proof}

For the case of the sphere $S_R^{n-1}$ we have the constraint function $F({\bf x})=\sum\limits_{i=1}^n x_i^2$ and, applying the formula \eqref{sigma-101}, the Lagrange multiplier function is 
\begin{equation*}
\sigma({\bf x})=\frac{\left<\nabla F({\bf x}),\nabla f({\bf x})\right>}{\|\nabla F({\bf x})\|^2}=\frac{1}{2R^2}\left<{\bf x}, \nabla f({\bf x})\right>,
\end{equation*}
where $f:\R^n\to \R$ is a smooth prolongation of a smooth function $\widetilde{f}:S_R^{n-1}\to \R$. 

We now have all the necessary ingredients in order to compute the Laplace-Beltrami operator for the function $\widetilde{f}$ applying formula \eqref{laplace-general-101}. 

\begin{thm}
For ${\bf x}\in S_R^{n-1}\subset \R^n$ we have the formula
\begin{equation}\label{Laplace-sfera-2022}
\left(\Delta_{S_R^{n-1}}\widetilde{f}\right)({\bf x})=\left(\Delta_{\R^n}f\right)({\bf x})-\frac{n-1}{R^2}\left<{\bf x}, \nabla f({\bf x})\right>-\frac{1}{R^2}{{\bf x}^t}\left[\emph{Hess}f\right]({\bf x}){\bf x}.
\end{equation}
\end{thm}

\begin{proof}
Embedding the result $(ii)$ from the Lemma \ref{lema1} in equation \eqref{laplace-general-101}, we successively get
\begin{align*}
\hspace*{-.2cm}\left(\Delta_{S_R^{n-1}}\widetilde{f}\right)({\bf x})&=\hspace*{-.1cm}\tr\left(\left(\mathbb{I}_n-\dfrac{1}{R^2}{\bf x}\,{\bf x}^t\right)\left[\text{Hess}f\right]({\bf x})\hspace*{0cm}\right)-\frac{\left<{\bf x}, \nabla f({\bf x})\right>}{2R^2}\tr\left(\left(\mathbb{I}_n-\dfrac{1}{R^2}{\bf x}\,{\bf x}^t\right)2\mathbb{I}_n\right) \\
&=\tr\left(\left[\text{Hess}f\right]({\bf x})\right)-\frac{\tr\left({{\bf x}\,{\bf x}^t}\left[\text{Hess}f\right]({\bf x})\right)}{R^2} -\frac{\left<{\bf x}, \nabla f({\bf x})\right>}{R^2}\tr\left(\mathbb{I}_n-\dfrac{1}{R^2}{\bf x}\,{\bf x}^t\right) \\
&=\left(\Delta_{\R^n} f\right)({\bf x})-\frac{n-1}{R^2}\left<{\bf x}, \nabla f({\bf x})\right>-\frac{1}{R^2}{{\bf x}^t}\left[\hbox{Hess}f\right]({\bf x}){\bf x}.
\end{align*}
\end{proof}


\noindent Using other arguments, the above result has been previously obtained in the literature, see \cite{atkinson}, \cite{bakry}, and \cite{bakry-2}. 

An interesting particular case is obtained when the prolongation function $f$ is homogeneous of degree $k$ (see \cite{shubin}, \cite{dai}). 

\begin{cor}
Let be $\widetilde{f}:S_R^{n-1}\to \R$ a smooth function and $f:\R^n\to \R$ a prolongation of $\widetilde{f}$, $f$ being a homogeneous\footnote{A function $f:\R^n\to \R$ is called  homogeneous of degree $k\in \mathbb{Z}$ if $f(t{\bf x})=t^kf({\bf x})$, for all $t\in \R^*$ and for all ${\bf x}\in \R^n$.} function of degree $k$. Then,
\begin{equation*}
\Delta_{S_R^{n-1}}\widetilde{f}=\left(\Delta_{\R^n}f\right)_{|S_R^{n-1}}-\frac{k(k+n-2)}{R^2}\widetilde{f}.
\end{equation*} 
\end{cor}

\begin{proof}
From Euler's theorem for homogeneous functions we have 
$\left<{\bf x}, \nabla f({\bf x})\right>=kf({\bf x})$ and $$\tr\left({{\bf x}\,{\bf x}^t}\left[\text{Hess}f\right]({\bf x})\right)=\sum\limits_{i,j=1}^n \frac{\partial^2 f}{\partial x_i\partial x_j}({\bf x})x_ix_j=k(k-1)f({\bf x}).$$ Substituting these equalities in \eqref{Laplace-sfera-2022} we obtain the desired result.
\end{proof}

\section{\texorpdfstring{Laplace-Beltrami operator on the orthogonal group}{Laplace-Beltrami operator on the orthogonal group}}

For a matrix $U\in  \mathcal{M}_{n}(\R)$, we denote by ${\bf u}_1,...,{\bf u}_n\in \R^n$ the  vectors formed with the columns of the matrix $U$ and consequently, $U$ has the form $U=\left[{\bf u}_1,...,{\bf u}_n\right]$. If $U\in O(n)=\{U\in \mathcal{M}_{n}(\R) \,|\,U^tU=UU^t=\mathbb{I}_n\}$, then the vectors ${\bf u}_1,...,{\bf u}_n\in \R^n$ are orthonormal.
We identify $\mathcal{M}_{n}(\R)$ with $\R^{n^2}$ by the isomorphism $\text{vec}:\mathcal{M}_{n}(\R)\rightarrow \R^{n^2}$ defined by the column vectorization $\text{vec}(U)\stackrel{\text{not}}{=}{\bf u}:=({\bf u}_1^t,...,{\bf u}_n^t)^t$.

The constraint functions $F_{aa},F_{bc}:\R^{n^2}\rightarrow \R$ that describe the orthogonal group as a preimage of a regular value are given by:
\begin{align*}
F_{aa} ({\bf u}) =\frac{1}{2}\|{\bf u}_a\|^2,\,\,1 \leq a\leq n;\,\,\,\,\,\,F_{bc} ({\bf u}) = \left<{\bf u}_b,{\bf u}_c\right>,\,\,1\leq b<c\leq n. 
\end{align*}
We have ${\bf F}:\R^{n^2}\rightarrow \R^{\frac{n(n+1)}{2}}$, ${\bf F}:=\left( \dots , F_{aa},\dots ,F_{bc}, \dots\right)$, $O(n)\simeq {\bf F}^{-1}\left( \dots , \frac{1}{2},\dots ,0, \dots\right)$ $\subset \R^{n^2}.$

The tangent space in a point $U\in O(n)$ is given by\footnote{$\text{Skew}_n(\R)=\{\Theta\in \mathcal{M}_n(\R)\,|\,\Theta^t=-\Theta\}$.}
$
\text{T}_UO(n)=\{U\Theta\,|\,\Theta\in \text{Skew}_n(\R)\}.
$
Considering the following basis for the $\frac{n(n-1)}{2}$-dimensional vector space 
$\text{Skew}_n(\R)$, 
$$\Theta_{ab}=(-1)^{a+b}({\bf e}_b{\bf e}_a^t-{\bf e}_a{\bf e}_b^t),\,\, 1\leq a<b\leq n,$$
we construct the local frame on $O(n)$ 
$$W_{ab}(U)=U\Theta_{ab},\,\,1\leq a<b\leq n.$$
The transformation matrix \eqref{transformation-matrix-11} in the point $U$ becomes
$$T=\left[{\bf t}_{12}, {\bf t}_{13},\dots, {\bf t}_{1n},{\bf t}_{23},\dots, {\bf t}_{2n},\dots, {\bf t}_{n-1,n}\right]\in \mathcal{M}_{n^2\times \frac{n(n-1)}{2}}(\R),$$
where
${\bf t}_{ab}(U)=\text{vec}(W_{ab}(U)).$

\begin{lem}\label{T^TT} The transformation matrix $T$ has the property 
$T^tT=2\mathbb{I}_{\frac{n(n-1)}{2}}.$
\end{lem}

\begin{proof}
For $U\in O(n)$ we compute the $(i,j)$ element of $T^tT$, $i,j\in \{1,\dots,\frac{n(n-1)}{2}\}$. The index $i$ corresponds to the vector $W_{ab}$ in the chosen local frame and the index $j$ corresponds to the vector $W_{cd}$. We obtain, see \cite{second-order}:
\begin{align*}
(T^tT)_{ij} & = \left<{\bf t}_{ab},{\bf t}_{cd}\right>==\tr(W_{ab}^tW_{cd}) 
 = \tr((U\Theta_{ab})^t(U\Theta_{cd}))=\tr(\Theta_{ab}^tU^tU\Theta_{cd})=\tr(\Theta_{ab}^t\Theta_{cd})
 = 2\delta_{ij}.
\end{align*}
\end{proof}

\begin{lem}\label{TT^t} For $U=\left[{\bf u}_1,...,{\bf u}_n\right]\in O(n)$ the transformation matrix $T$ has the property
$$
TT^t=\mathbb{I}_{n^2}-\left[\begin{array}{c|c|c}
\textbf{u}_1 \textbf{u}_1^t&\dots&\textbf{u}_n \textbf{u}_1^t\\
\hline
\vdots&\ddots&\vdots\\
\hline
\rule{0pt}{12pt}\textbf{u}_1 \textbf{u}_n^t&\dots&\textbf{u}_n \textbf{u}_n^t\end{array}\right].
$$
\end{lem}

\begin{proof} Denote by $I=\{(a,b)\in \{1,2,\dots,n\}\times \{1,2,\dots,n\}|~a<b\}$, the set being lexicographically  ordered.
We organize the matrix $T$ as a column stack of the matrices $S_i$, $i\in \{1,2,\dots,n\}$, each of these matrices having dimension $n\times \frac{n(n-1)}{2}$. The matrix $S_i$ is formed with the $i^{th}$ columns of the matrices $W_{ab}$, with $(a,b)\in I$. More precisely, we have
$$T=\left[\begin{array}{c}
S_1\\
\hline
\vdots\\
\hline
S_n
\end{array}\right];~~~T^t=\left[\begin{array}{c|c|c}S_1^t&\dots&S_n^t
\end{array}\right].$$
Then, the product $TT^t$ can be written in a $n\times n$ blocks form, with the $(i,j)$ block being the $n\times n$ product matrix $S_iS_j^t$, $i,j\in \{1,2,\dots,n\}$. \\
Next, we notice that the $(p,q)$ element in the matrix $W_{rs}$ is given by
$$
(-1)^{r+s}\sum_{v=1}^n u_{pv}\cdot \left(\textbf{e}_s\textbf{e}_r^t-\textbf{e}_r \textbf{e}_s^t\right)_{vq}=(-1)^{r+s}\sum_{v=1}^nu_{pv}\left(\delta_{vr}\delta_{qs}-\delta_{vs}\delta_{qr}\right)=(-1)^{r+s}\left(u_{pr}\delta_{qs}-u_{ps}\delta_{qr}\right).$$
We are now ready to compute {the} ${(k,l)}$ {element} {of the block} ${S_iS_j^t}$, $i,j,k,l\in \{1,\dots,n\}$, as being the dot product of the $k^{th}$ row of the matrix $S_i$ with the $l^{th}$ row of the matrix $S_j$, taking into account that the first mentioned row is formed with the $k^{th}$ row element in the $i^{th}$ column of {all} matrices $W_{ab},(a,b)\in I$, and that the second mentioned row is formed with the $l^{th}$ row element in the $j^{th}$ column of {all} matrices $W_{ab},(a,b)\in I$:
\begin{align*}
(S_iS_j^t)_{kl} &=\sum_{(a,b)\in I}(-1)^{a+b}\left(u_{ka}\delta_{ib}-u_{kb}\delta_{ia}\right)\cdot (-1)^{a+b}\left(u_{la}\delta_{jb}-u_{lb}\delta_{ja}\right)\\
&=\sum_{(a,b)\in I}\left(u_{ka}\delta_{ib}u_{la}\delta_{jb}-u_{ka}\delta_{ib}u_{lb}\delta_{ja}-u_{kb}\delta_{ia}u_{la}\delta_{jb}+u_{kb}\delta_{ia}u_{lb}\delta_{ja}\right)\\
&=\left(\sum_{(a,b)\in I}u_{ka}u_{la}\delta_{ib}\delta_{jb}+\sum_{(a,b)\in I}u_{kb}u_{lb}\delta_{ia}\delta_{ja}\right)-\left(\sum_{(a,b)\in I}u_{ka}u_{lb}\delta_{ib}\delta_{ja}+\sum_{(a,b)\in I}u_{kb}u_{la}\delta_{ia}\delta_{jb}\right)\\
&=\sum_{a\ne b}u_{ka}u_{la}\delta_{ib}\delta_{jb}-\sum_{a\ne b}u_{ka}u_{lb}\delta_{ib}\delta_{ja}.
\end{align*}
For $i\ne j$ the first of the above two terms vanishes, while the second term is $-u_{kj}u_{li}$. For $i=j$ the second of the above two terms vanishes, while the first term is
$$\sum_{a\ne i}u_{ka}u_{la}=\left(\sum_{a=1}^nu_{ka}u_{la}\right)-u_{ki}u_{li}=\delta_{kl}-u_{ki}u_{li}.$$
Therefore, we have obtained that
$$(S_iS_j^t)_{kl}=\left\{\begin{array}{ll}
\delta_{kl}-u_{ki}u_{li},&\hbox{if}~~i=j\\
-u_{kj}u_{li}, &\hbox{if}~~i\ne j\end{array}\right. =\delta_{ij}\delta_{kl}-u_{kj}u_{li},$$
which completely proves the result. 
\end{proof}

We introduce the following notation:
$$\Lambda(U):=\left[\begin{array}{c|c|c}
\textbf{u}_1 \textbf{u}_1^t&\dots&\textbf{u}_n \textbf{u}_1^t\\
\hline
\vdots&\ddots&\vdots\\
\hline
\rule{0pt}{12pt}\textbf{u}_1 \textbf{u}_n^t&\dots&\textbf{u}_n \textbf{u}_n^t\end{array}\right].$$
Note that $\Lambda(U)$ is an involution matrix $\left((\Lambda(U))^2=\mathbb{I}_{n^2}\right)$, having the eigenvalues 1 (of multiplicity $\dfrac{n(n+1)}{2}$) and -1 (of multiplicity $\dfrac{n(n-1)}{2}$).

For a smooth function $f:\mathcal{M}_n(\R)\to \R$ we denote $\widehat{f}:\R^{n^2}\to \R$, $\widehat{f}:=f\circ\text{vec}^{-1}$ and
$$\nabla f(U):=\text{vec}^{-1}\left(\nabla_{\text{Euc}}\widehat{f}(\text{vec}(U))\right),\,\,\Delta f(U):=\Delta_{\text{Euc}}\widehat{f}(\text{vec}(U)),$$
$$[\text{Hess}f](U):=[\text{Hess}\widehat{f}](\text{vec}(U))=\left[\begin{array}{ccc}
\left[\frac{\partial^2 \widehat{f}(\text{vec}(U))}{\partial {\bf u}_1\partial {\bf u}_1}\right] &\dots&\left[\frac{\partial^2 \widehat{f}(\text{vec}(U))}{\partial {\bf u}_1\partial {\bf u}_n}\right]\\
\vdots&\ddots&\vdots\\
\rule{0pt}{12pt}\left[\frac{\partial^2 \widehat{f}(\text{vec}(U))}{\partial {\bf u}_n\partial {\bf u}_1}\right]&\dots&\left[\frac{\partial^2 \widehat{f}(\text{vec}(U))}{\partial {\bf u}_n\partial {\bf u}_n}\right]\end{array}\right],$$
where
$$\left[\frac{\partial^2 \widehat{f}}{\partial {\bf u}_i\partial {\bf u}_j}\right]:=\left[\begin{array}{ccc}
\frac{\partial^2 \widehat{f}}{\partial {u}_{1i}\partial {u}_{1j}} &\dots&\frac{\partial^2 \widehat{f}}{\partial {u}_{1i}\partial {u}_{nj}}\\
\vdots&\ddots&\vdots\\
\rule{0pt}{12pt}\frac{\partial^2\widehat{f}}{\partial {u}_{ni}\partial {u}_{1j}}&\dots&\frac{\partial^2 \widehat{f}}{\partial {u}_{ni}\partial {u}_{nj}}\end{array}\right].$$

On the orthogonal group $O(n)$ we consider the bi-invariant metric induced by the Frobenius metric on $\mathcal{M}_n(\R)$ and the Laplace-Beltrami operator on $O(n)$ is computed with respect to this metric.

The following result provides a formula for the Laplace-Beltrami operator on the orthogonal group, written in the ambient (Euclidean) coordinates of $\mathcal{M}_n(\R)$.

\begin{thm}\label{teorema-principala}
Let be $\widetilde{f}:O(n)\to \R$ a smooth function and $f:\mathcal{M}_n(\R)\to \R$ a smooth prolongation of $\widetilde{f}$. Then, for $U\in O(n)$, we have 
\begin{equation*}
\Delta_{O(n)}\widetilde{f}(U)=\frac{1}{2}\Delta f(U)-\frac{n-1}{2}\tr\left(U^t[\nabla f](U)\right)-\frac{1}{2}\tr\left(\Lambda(U)[\emph{Hess}f](U)\right).
\end{equation*}
\end{thm}

\begin{proof}
For $O(n)$ we have, see \cite{second-order}:
$$[\text{Hess}_{O(n)}\widetilde{f}](U)=\left([\text{Hess}f](U)-\Sigma(U)\otimes \mathbb{I}_n\right)_{|T_UO(n)\times T_UO(n)},$$
where (\cite{first-order}, pp. 1779)
$$\Sigma(U)=\frac{1}{2}\left([\nabla f]^t(U)U+U^t[\nabla f](U)\right)$$
and, by Lemma \ref{T^TT}, 
$[{\bf g}_c]^{-1}=(T^tT)^{-1}=\frac{1}{2}\mathbb{I}_{\frac{n(n-1)}{2}}$ and therefore $T^+=\frac{1}{2}T^t$. We denote the elements of the matrix $\Sigma$ by $\sigma_{ij},~i,j\in \{1,\dots,n\}$.

Using Lemma \ref{TT^t}, the formula \eqref{laplace-general-101} from Theorem \ref{teorema-generala} becomes
\begin{align*}
\Delta_{O(n)}\widetilde{f} & = \tr\left(T^+\left([\text{Hess}f]-\Sigma\otimes \mathbb{I}_n\right)T\right) = \frac{1}{2}\tr\left(TT^t\left([\text{Hess}f]-\Sigma\otimes \mathbb{I}_n\right)\right) \\
& =  \frac{1}{2}\tr\left(\left(\mathbb{I}_{n^2}-\Lambda\right)\left([\text{Hess}f]-\Sigma\otimes \mathbb{I}_n\right)\right) 
 =  \frac{1}{2}\tr\left([\text{Hess}f]-\Sigma\otimes \mathbb{I}_n-\Lambda [\text{Hess}f]+\Lambda (\Sigma\otimes \mathbb{I}_n)\right) \\
& = \frac{1}{2}\Delta {f}-\frac{1}{2}\tr(\Sigma)\tr(\mathbb{I}_n)-\frac{1}{2}\tr(\Lambda [\text{Hess}f])+\frac{1}{2}\tr(\Lambda (\Sigma\otimes \mathbb{I}_n)) \\
& = \frac{1}{2}\Delta {f}-\frac{n}{2}\tr(\Sigma)-\frac{1}{2}\tr(\Lambda [\text{Hess}f])+\frac{1}{2}\tr(\Lambda (\Sigma\otimes \mathbb{I}_n)).
\end{align*}
From the above expression of the matrix $\Sigma$ we obtain $\tr(\Sigma(U))=\tr\left(U^t[\nabla f](U)\right)$. Also, we have 
$$\Lambda(U) (\Sigma\otimes \mathbb{I}_n)=\left[\begin{array}{c|c|c}
\textbf{u}_1 \textbf{u}_1^t&\dots&\textbf{u}_n \textbf{u}_1^t\\
\hline
\vdots&\ddots&\vdots\\
\hline
\rule{0pt}{7pt}\textbf{u}_1 \textbf{u}_n^t&\dots&\textbf{u}_n \textbf{u}_n^t\end{array}\right]\cdot \left[\begin{array}{c|c|c}
\sigma_{11}\mathbb{I}_n&\dots& \sigma_{1n}\mathbb{I}_n\\
\hline
\vdots&\ddots&\vdots\\
\hline
\rule{0pt}{12pt}\sigma_{n1}\mathbb{I}_n&\dots&\sigma_{nn}\mathbb{I}_n\end{array}\right].
$$
The diagonal block of position $(k,k)$ in the above block matrix multiplication, with $k\in \{1,\dots,n\}$, is
$\sum\limits_{i=1}^n\sigma_{ik}\textbf{u}_i \textbf{u}_k^t$.
Therefore, using the orthogonality constraints, its trace is $$\sum_{i=1}^n\sigma_{ik}\tr(\textbf{u}_i \textbf{u}_k^t)
=\sum_{i=1}^n\sigma_{ik}\tr(\textbf{u}_k^t \textbf{u}_i)=\sum_{i=1}^n\sigma_{ik}(\textbf{u}_k^t \textbf{u}_i)=\sum_{i=1}^n\sigma_{ik}\delta_{ik}=\sigma_{kk}.
$$
It follows that
$$\tr(\Lambda(U) (\Sigma\otimes \mathbb{I}_n))=\tr(\Sigma),$$
which concludes the proof.
\end{proof}

In the next two subsections, we apply the formula for the Laplacian on the orthogonal group for a few functions of interest defined on $O(n)$.
\subsection{Laplace-Beltrami operator for some power sum symmetric functions}

In the context of analyzing the heat kernel on the orthogonal group, the Laplacian of certain power sum symmetric functions defined on the orthogonal group needs to be computed, see \cite{levy}, \cite{fulman}. More precisely, these functions are
$$\widetilde p_1(U)=\tr(AU);~\widetilde p_{1,1}(U)=(\tr(AU))^2;~\widetilde p_2(U)=\tr((AU)^2),$$
where $A$ is a given matrix from $\mathcal{M}_{n}(\R)$. We denote by $p_1,p_{1,1},p_2:\mathcal{M}_{n}(\R)\to \R$ their natural prolongations.
By the Riesz representation theorem, any linear function on $\mathcal{M}_{n}(\R)$ is of the form $p_1$.
 
We have $[\nabla p_1](U)=A^t$, $\Delta p_1(U)=0$, and $[\text{Hess}\,p_1](U)=O_{n^2}$. From Theorem \ref{teorema-principala}, we obtain, for $U\in O(n)$,
$$\Delta_{O(n)}\widetilde{p}_1(U)=-\frac{n-1}{2}\wt{p}_1(U).$$
This equality shows that the restriction to $O(n)$ of the linear function $p_1$ is an eigenvector for the Laplace-Beltrami operator, with the associated eigenvalue $-\frac{n-1}{2}$.

We proceed with the computations for the Laplacian of the function $\widetilde{p}_{1,1}$. We have $[\nabla p_{1,1}](U)=2\tr(AU)A^t$. If we denote $B:=A^t$, then we straightforwardly obtain $[\text{Hess}\,p_{1,1}](U)=2\text{vec}(B)\text{vec}(B)^t$. Therefore, $\Delta p_{1,1}(U)=2\tr\left(\text{vec}(B)\text{vec}(B)^t\right)=2\text{vec}(B)^t\text{vec}(B)=2\tr(B^tB)=2\tr(AA^t)$. If we denote the columns of $B$ by ${\bf b}_1,\dots,{\bf b}_n$, then 
$$[\text{Hess}\,p_{1,1}](U)=2\left[\begin{array}{c|c|c}
\textbf{b}_1 \textbf{b}_1^t&\dots&\textbf{b}_1 \textbf{b}_n^t\\
\hline
\vdots&\ddots&\vdots\\
\hline
\rule{0pt}{12pt}\textbf{b}_n \textbf{b}_1^t&\dots&\textbf{b}_n \textbf{b}_n^t\end{array}\right].$$
For $k\in \{1,\dots,n\}$ the $(k,k)$ block matrix of the product $\Lambda[\hbox{Hess}~p_{1,1}]$ is
$$2\sum\limits_{j=1}^n \left({\bf u}_j{\bf u}_k^t\right)\left({\bf b}_j{\bf b}_k^t\right)=2
\sum\limits_{j=1}^n\left({\bf u}_k^t{\bf b}_j\right)\left({\bf u}_j{\bf b}_k^t\right),$$
therefore 
$$\tr\left(\Lambda[\hbox{Hess}~p_{1,1}]\right)=2\sum_{k=1}^n\sum_{j=1}^n\left({\bf u}_k^t{\bf b}_j\right)\left({\bf u}_j^t{\bf b}_k\right).$$
On the other side, we have
$AU=\left[{\bf u}_j^t{\bf b}_i\right]_{i,j}$ and, for $k\in \{1,\dots,n\}$, the $(k,k)$ element of the matrix $(AU)^2$ is
$\sum\limits_{j=1}^n\left({\bf u}_j^t{\bf b}_k\right)\left({\bf u}_k^t{\bf b}_j\right),$
hence
$${p}_2(U)=\tr((AU)^2)=\sum_{k=1}^n\sum_{j=1}^n\left({\bf u}_j^t{\bf b}_k\right)\left({\bf u}_k^t{\bf b}_j\right)=\frac{1}{2}\cdot \tr\left(\Lambda[\hbox{Hess}~p_{1,1}]\right).$$
Substituting in the formula from Theorem 3.3 we obtain, for $U\in O(n)$, that
\begin{align*}
\Delta_{O(n)}\widetilde{p}_{1,1}(U)&=\tr(AA^t)-(n-1)\tr(U^t\tr(AU)A^t)-
\widetilde{p}_2(U)\\
&=\tr(AA^t)-(n-1)\tr(AU)\tr(U^tA^t)-\widetilde{p}_2(U)\\
&=\tr(AA^t)-(n-1)\widetilde{p}_{1,1}(U)-\widetilde{p}_2(U).
\end{align*}

We proceed now with the computations for the Laplacian of the function $\widetilde{p}_{2}$. We have $[\nabla p_{2}](U)=2A^tU^tA^t$. We straightforwardly obtain 
$$[\text{Hess}\,p_{2}](U)=2\left[\begin{array}{c|c|c}
\textbf{b}_1 \textbf{b}_1^t&\dots&\textbf{b}_n \textbf{b}_1^t\\
\hline
\vdots&\ddots&\vdots\\
\hline
\rule{0pt}{12pt}\textbf{b}_1 \textbf{b}_n^t&\dots&\textbf{b}_n \textbf{b}_n^t\end{array}\right].$$
Therefore, $\Delta p_{2}(U)=\tr\left([\text{Hess}\,p_{2}](U)\right)=2\tr(AA^t)$. For $k\in \{1,\dots,n\}$ the $(k,k)$ block matrix of the product $\Lambda[\hbox{Hess}~p_2]$ is
$2\sum\limits_{j=1}^n \left({\bf u}_j{\bf u}_k^t\right)\left({\bf b}_k{\bf b}_j^t\right),$
hence $$\tr\left(\Lambda[\hbox{Hess}~p_2(U)]\right)=2\sum_{k=1}^n\sum_{j=1}^n \left({\bf u}_k^t{\bf b}_k\right)\left({\bf u}_j^t{\bf b}_j\right)=2\left(\sum_{k=1}^n{\bf u}_k^t{\bf b}_k\right)^2=2\left(\tr(AU)\right)^2=2p_{1,1}(U).$$
Substituting in the formula from Theorem 3.3 we obtain, for $U\in O(n)$, that
\begin{align*}
\Delta_{O(n)}\widetilde{p}_{2}(U)&=\tr(AA^t)-(n-1)\tr(U^tA^tU^tA^t)-\widetilde{p}_{1,1}(U)\\
&=\tr(AA^t)-(n-1)\tr((AU)^2)-\widetilde{p}_{1,1}(U)\\
&=\tr(AA^t)-(n-1)\widetilde{p}_{2}(U)-\widetilde{p}_{1,1}(U).
\end{align*}
These equalities show that the restriction to $O(n)$ of the function $p_{1,1}-p_2$ is an eigenvector for the Laplace-Beltrami operator, with the associated eigenvalue $-(n-2)$.

\subsection{Laplace-Beltrami operator for the Brockett function}

We consider the Brockett function $\widetilde{G}:O(n)\to \R$, 
$$\wt{G}(U)=\tr(U^tAUN),$$
with $A,N\in \mathcal{M}_{n}(\R)$, $A=A^t$ and $N=\text{diag}(\mu_1,\dots,\mu_n)$, which appears in the well-known Brockett optimization problem (see \cite{brockett}, \cite{first-order}, \cite{second-order}). Its natural prolongation $G:\mathcal{M}_n(\R)\to \R$ is given by
$G(U)=\tr(U^tAUN)$.
We have, see \cite{first-order}, \cite{second-order}, that 
$$[\nabla G](U)=2AUN,\,\,\,[\text{Hess}\,G](U)=2N\otimes A.$$
Also,  
$\Delta G(U)=2\tr (N)\tr(A)$,
$\tr(U^t[\nabla G](U))=2G(U)$, and
$$\Lambda(U)[\text{Hess}\,G](U)=2\begin{pmatrix}
\mu_1({\bf u}_1{\bf u}_1^t)A & * & \dots & * \\
* & \mu_2({\bf u}_2{\bf u}_2^t)A & \dots & * \\
\vdots & \vdots & \ddots & \vdots \\
* & * & \dots & \mu_n({\bf u}_n{\bf u}_n^t)A
\end{pmatrix}.$$
Therefore 
$$\tr(\Lambda(U)[\text{Hess}\,G](U))=2\tr((\mu_1{\bf u}_1{\bf u}_1^t+\dots +\mu_n{\bf u}_n{\bf u}_n^t)A).$$
For $U\in O(n)$ we obtain the formula for the Laplace-Beltrami operator of the Brockett function
$$\Delta_{O(n)}\wt{G}(U)=-(n-1)\wt{G}(U)+\tr(N)\tr(A)-\tr((\mu_1{\bf u}_1{\bf u}_1^t+\dots +\mu_n{\bf u}_n{\bf u}_n^t)A).$$

\medskip

\noindent {\bf Acknowledgments.} We are thankful to our friend and colleague \'Arp\'ad B\'enyi for the useful discussions about the manuscript.

\end{document}